\documentclass[aps,pra,reprint,nofootinbib,groupedaddress,twoside,showpacs]{revtex4-1}

\usepackage[utf8]{inputenc}
\usepackage{amsmath,amssymb}
\usepackage{amsthm}
\usepackage{hyperref} 
\usepackage{acronym}  
\usepackage{color}  
\usepackage{amsfonts}
\usepackage{mathrsfs}
\usepackage{graphicx}
\usepackage[on]{psfrag}
\usepackage{array}
\usepackage{cases}
\usepackage{tabularx}
\usepackage{url}
\usepackage{soul}
\usepackage{theoremref}
\usepackage{thmtools}
\usepackage{lettrine}
\usepackage{fancyhdr}
\usepackage{WinsNotation}

\makeatletter
\renewcommand*\env@matrix[1][\arraystretch]{%
  \edef\arraystretch{#1}%
  \hskip -\arraycolsep
  \let\@ifnextchar\new@ifnextchar
  \array{*\c@MaxMatrixCols c}}
\makeatother

\newcommand*{\QEDB}{\hfill\ensuremath{\square}}%

\newcommand{\Ps}{\mathscr{P}}

\declaretheoremstyle[
spaceabove=8pt, spacebelow=8pt,
headfont=\normalfont\bfseries,
notefont=\normalfont, notebraces={(}{)},
headpunct={:},
bodyfont=\normalfont,
postheadspace=0.5em,
]{mynote}
\declaretheorem[style=mynote]{Remark}
\declaretheorem[style=mynote,name=Theorem]{Thm}
\declaretheorem[style=mynote,name=Lemma]{Lem}

\declaretheorem[style=mynote,name=Algorithm]{Algno}

\newcommand{\paperTitle}{Efficient Entanglement Distillation for Quantum Channels\\ with Polarization Mode Dispersion}

\begin{document}
{\color{white} 
\fontsize{0pt}{0pt}\selectfont
\begin{acronym}
\acro{SVD}{singular value decomposition}\vspace{-16mm}
\acro{LOCC}{local operations and classical communication}\vspace{-16mm}
\acro{QED}{quantum entanglement distillation}\vspace{-16mm}
\acro{QEC}{quantum error correction}\vspace{-16mm}
\acro{w.r.t.}{with respect to}\vspace{-16mm}
\acro{FP}{fidelity-prioritized}\vspace{-16mm}
\acro{PP}{probability-prioritized}\vspace{-16mm}
\acro{QPA}{quantum privacy amplification}\vspace{-16mm}
\acro{TKO}{two-Kraus-operator}\vspace{-16mm}
\acro{PMD}{polarization mode dispersion}\vspace{-16mm}
\acro{PSP}{principal states of polarization}\vspace{-16mm}
\acro{DGD}{differential group delay}\vspace{-16mm}
\acro{DFS}{decoherence-free subspace}\vspace{-16mm}
\acro{CW}{continuous-wave}\vspace{-16mm}
\end{acronym}}


\setcounter{page}{1}
\title{\paperTitle}

\author{Liangzhong~Ruan, Moe~Z.~Win}
\affiliation{Massachusetts Institute of Technology, Cambridge, MA}
\author{\vspace{-2mm} Brian T. Kirby, and Michael Brodsky}
\affiliation{U.S. Army Research Laboratory, Adelphi, MD}



\date{\today}

\begin{abstract}
Quantum entanglement shared by remote network nodes serves as a valuable resource for promising applications in distributed computing, cryptography, and sensing. 
However, distributing entanglement with high-quality via fiber optic routes could be challenging due to the various decoherence mechanisms in fibers. 
In particular, one of the primary polarization decoherence mechanisms in optical fibers is \ac{PMD}, which is the distortion of optical pulses by randomly varying birefringences in the system.
To  mitigate the effect of decoherence in entangled particles, \ac{QED} algorithms have been proposed.
One particular class, the recurrence \ac{QED} algorithms, stands out because it has relatively relaxed requirements on both the size of the quantum circuits involved and on the initial quality of entanglement in particles. 
However, because the number of particles required grows exponentially with the number of rounds of distillation, an efficient recurrence algorithm needs to converge quickly.
We present a recurrence \ac{QED} algorithm designed for photonic qubit pairs affected by \ac{PMD}-degraded channels.
Our proposed algorithm achieves the optimal fidelity 
as well as the optimal success probability (conditional on the optimal fidelity being achieved) in every round of distillation.  
The attainment of optimal fidelity improves the convergence speed of fidelity with respect to the rounds of distillation from linear to quadratic, and hence significantly reduces the number of distillation rounds.
Combined with the fact that the optimal success probability is achieved, 
the proposed algorithm provides an efficient method to distribute entangled states with high fidelity via optic fibers.
\end{abstract}


\maketitle

\acresetall             

\section{Introduction}

Applications of quantum information protocols, such as teleportation \cite{BenBraCreJozPerWoo:93,NieKniLaf:98,GotChu:99},  dense coding \cite{BenSte:92,WanDenLiLiuLon:05,BarWeiKwi:08}, entanglement-assisted quantum key distribution \cite{Eke:91,KoaPre:03,GotLoLutPre:04}, and quantum repeaters \cite{DurBriCirZol:99,SanSimRieGis:11,KirSanMalBro:16}, rely on the ability of distributing quantum entanglement among distant network nodes, a task for which the fiber-optic infrastructure is a natural candidate.
In the context of delivering entanglement, 
polarization-entangled photon pairs are particularly useful because of the ease with which light polarization can be manipulated using standard instrumentation \cite{PopFedUrsBohLorMauPeevSudKurWei:OE04} and the numerous sources of polarization-entangled photons suitable for use with standard fibers \cite{WanKan:IEEE09}. 
 For polarization-entangled photons, the major decoherence mechanism is birefringence \cite{AntShtBro:PRL11,ShtAntBro:OE11,BroGeoAntSht:OL11}. 
 The accumulation of randomly varying birefringence in fibers leads to a phenomenon known as \ac{PMD} \cite{GorKog:PNAS00}.
 
In the literature, the \ac{PMD} effect is often modeled using the first-order approximation \cite{AntShtBro:PRL11,ShtAntBro:OE11}.
As illustrated in Fig.~\ref{fig_PMD}, with this approximation, the \ac{PMD} effect is parametrized by \ac{PSP} and \ac{DGD}, both of which vary stochastically in time.  
However, since typical time constants characterizing the decorrelation of \ac{PMD} in buried optical fibers are as long as hours, days and sometimes months \cite{BroFriBorTur:JLT06}, \ac{PMD} evolution can be considered adiabatic in the context of quantum communications protocols.
Therefore, it is reasonable to assume that the parameters of the \ac{PMD} effect, particularly \ac{PSP}, can be measured by the network nodes.

\begin{figure*}[t] \centering
\includegraphics[scale=0.45]{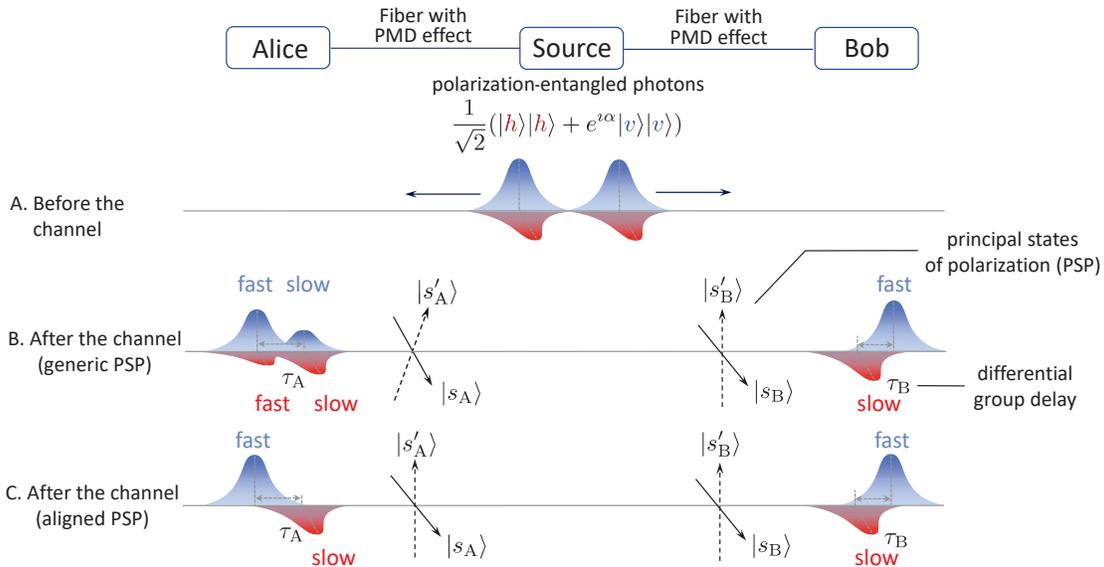}
\caption {System Model. With first-order approximation, the overall effect of \ac{PMD} resembles that of pure birefringence in the sense that it causes an incident pulse to split into two orthogonally polarized components delayed relative to each other \cite{GorKog:PNAS00}. 
The polarization states of these two components are known as the \ac{PSP} and the delay between them is called the \ac{DGD}.
Appendix~\ref{pf_align} shows that even with generic \ac{PSP}, a maximally entangled polarization state
 prepared by the source can be viewed as if the polarization basis of one of the photons is already aligned with the \ac{PSP} basis of the channel.
 Hence, in this figure, the polarization basis of photon B is always aligned with the \ac{PSP} basis of the channel.}
\label{fig_PMD}
\end{figure*}

As illustrated in Fig.~\ref{fig_PMD}A, to deliver entanglement to remote network nodes, Alice and Bob, the source locally generates a maximally entangled photon pair and respectively sends the two photons to the two nodes.
However, the decoherence effect of the channel deteriorates the entanglement during the transmission.
To address this problem, \ac{QED} algorithms \cite{BenBraPopSchSmoWoo:96,DeuEkeJozMacPopSan:96,BriDurCirZol:98,OpaKur:99,MunOrs:09,DehVanDeMVer:03,VolVer:05,HosDehDeM:06,BenDivSmoWoo:96,Mat:03,AmbGot:06,WatMatUye:06,Rozetal:18} have been proposed to generate qubit pairs in the targeted entangled state using \ac{LOCC}.
Since high-quality entanglement is the keystone in many  important applications of quantum computation and quantum information, \ac{QED} has become an essential building block for the development of quantum networks \cite{NicFitBenSim:14,Komar:14,Kaletal:17}.

Three types of \ac{QED} algorithms have been proposed in the literature, namely, asymptotic  \cite{DehVanDeMVer:03,VolVer:05,HosDehDeM:06}, code-based \cite{Mat:03,AmbGot:06,WatMatUye:06}, and recurrence algorithms \cite{BenBraPopSchSmoWoo:96,DeuEkeJozMacPopSan:96,BriDurCirZol:98,OpaKur:99,MunOrs:09}. 
Among the three types of algorithms, the recurrence ones require local operations on just one or two qubits, and are robust against severe decoherence.
The recurrence algorithms  operate on two qubit pairs each time, improving the quality of entanglement in one pair at the expense of the other pair, which is then discarded.
The algorithms keep repeating this operation to progressively increase the fidelity of the kept qubit pairs \ac{w.r.t.} the targeted entangled state.
These algorithms can mitigate the effect of stronger decoherence by performing more rounds of distillations.
In fact, the recurrence algorithm proposed in  \cite{BenBraPopSchSmoWoo:96} can distill contaminated qubit pairs into
maximally entangled qubit pairs as long as the initial fidelity of the contaminated qubit pairs \ac{w.r.t.} the targeted state is greater than $0.5$.
In \cite{HorHor:96,HorHorHor:97}, it has been proven that a state of qubit pairs is distillable if and only if its fidelity \ac{w.r.t.}  a certain maximally entangled state is greater than $0.5$.
To summarize, recurrence algorithms are preferable in terms of both implementability and robustness.


Despite their advantages, recurrence algorithms do have a drawback in terms of efficiency.
The efficiency of \ac{QED} algorithms is measured in terms of \emph{yield}, which is defined as the ratio between the number of highly entangled output qubit pairs and the number of input qubit pairs impaired by decoherence effects.
Since at least half of the entangled qubit pairs are discarded in each round of distillation, the efficiency of the recurrence algorithms decreases exponentially with the number of rounds.
To reduce the required rounds of distillation, one needs to design the \ac{LOCC} adopted in the algorithms so that the fidelity of the kept qubit pairs quickly approaches 1 \ac{w.r.t.} the rounds of distillation.
To achieve this objective, the \ac{QPA} algorithm was proposed in \cite{DeuEkeJozMacPopSan:96}, and was shown numerically to require fewer rounds of distillation than the algorithm in \cite{BenBraPopSchSmoWoo:96} for qubit pairs impaired by a quantum depolarizing channel.
However, performance of the \ac{QPA} algorithm was not characterized analytically.
In fact, a set of initial states was found in \cite{OpaKur:99} for which the \ac{QPA} algorithm was less efficient than the algorithm in \cite{BenBraPopSchSmoWoo:96}.
In \cite{OpaKur:99}, the design of distillation operations was formulated into an optimization problem,
which was inherently non-convex, and consequently, the optimal solution was not found.
Therefore, the issue of improving the efficiency of recurrence \ac{QED} algorithms remains an interesting challenge.

In this work, we report an efficient recurrence \ac{QED} algorithm for entangled photons impaired by the \ac{PMD} effect.
We envision that a key enabler for designing efficient recurrence \ac{QED} algorithms is to make them adaptive to the key parameters of \ac{PMD}.
Intuitively, compared to general algorithms,
\ac{QED} algorithms that adapt to channel-specific decoherence effects will better mitigate such effects and hence distill more efficiently.
In fact, it has been observed that knowing the channel benefits the performance of quantum error recovery \cite{FleShoWin:J07}, and channel-adaptive \ac{QEC} schemes that outperform classical ones \cite{FleShoWin:J08,FleShoWin:J08a} have been designed.
In the following, we will first analyze the effect of \ac{PMD} on photon pairs affected by \ac{PMD}-degraded channels,
then characterize the optimal fidelity and the optimal success probability that can be achieved via \ac{LOCC} in each round of distillation,
and finally design an algorithm to achieve the optimal fidelity and success probability.
By achieving the optimal fidelity and success probability, the proposed algorithm provides an efficient method to distribute entangled photons with high fidelity through quantum channels impaired by fiber birefringence.

\begin{widetext}
\begin{align}
&\V{\rho}=\frac{1}{2}\left[
\begin{array}{c@{\;}c@{\;}c@{\;}c}
 \vert \eta_{1}\vert^{2} & -\eta_{1}\eta_{2}e^{-\imath\alpha}R^{\dag}(\tau_{\mathrm{A}},0) & \eta_{1}\eta^{\dag}_{2}R^{\dag}(0,\tau_{\mathrm{B}}) & \eta^{2}_{1}e^{-\imath\alpha}R^{\dag}(\tau_{\mathrm{A}},\tau_{\mathrm{B}}) \\
  -\eta_{1}^{\dag}\eta_{2}^{\dag}e^{\imath\alpha}R(\tau_{\mathrm{A}},0)  & \vert \eta_{2}\vert^{2} &  
  -(\eta^{*}_{2})^{2}e^{\imath\alpha}R(\tau_{\mathrm{A}},-\tau_{\mathrm{B}})  
   &  -\eta_{1}\eta_{2}^{\dag}R^{\dag}(0,\tau_{\mathrm{B}})  \\
  \eta_{1}^{\dag}\eta_{2}R(0,\tau_{\mathrm{B}})  &  
    -(\eta_{2})^{2}e^{-\imath\alpha}R^{\dag}(\tau_{\mathrm{A}},-\tau_{\mathrm{B}})
  & \vert \eta_{2}\vert^{2} &  \eta_{1}\eta_{2}e^{-\imath\alpha}R^{\dag}(\tau_{\mathrm{A}},0)  \\
 (\eta_{1}^{\dag})^{2}e^{\imath\alpha}R(\tau_{\mathrm{A}},\tau_{\mathrm{B}})  &  -\eta_{1}^{\dag}\eta_{2}R(0,\tau_{\mathrm{B}})  & \eta_{1}^{\dag}\eta_{2}^{\dag}e^{\imath\alpha}R(\tau_{\mathrm{A}},0)  & \vert \eta_{1}\vert^{2} \\
\end{array}
\right]
\label{eqn:Gen_Density}
\end{align}
\end{widetext}

\noindent{\bf Organization}: 
Section~\ref{sec:system_prob} presents the system model and defines the optimization problems for recurrence \ac{QED} algorithms. 
Section~\ref{sec:upper_alg} characterizes the optimal performances of the defined problems and then designs a recurrence \ac{QED} algorithm that achieves the characterized  optimal performance.
Section~\ref{sec:numerical} provides several numerical tests for the proposed algorithm.
Finally, Section~\ref{sec:conclusion} gives the conclusion.

\noindent{\bf Notations:} $a$, $\V{a}$, and $\M{A}$ represent scalar, vector, and matrices, respectively.
$\mathrm{pha}\{\cdot\}$ denotes the phase of a complex number.
$(\cdot)^\dag$, $\mathrm{rank}\{\cdot\}$, $\det\{\cdot\}$  and $\mathrm{tr}\{\cdot\}$, denote the Hermitian transpose, rank, determinant, and trace of a matrix, respectively. $\mathrm{tr}_{i,j}\{\cdot\}$ denotes the partial trace \ac{w.r.t.} to the $i$-th and $j$-th qubits in the system. 
$\propto$ denotes the proportional relationship.
$\mathbb{I}_n$ denotes the $n\times n$ identity matrix, and $\imath$ is the unit imaginary number.

\section{System Model and Problem Formulation}
\label{sec:system_prob}
This section presents the system model and then defines the optimization problems for recurrence \ac{QED} algorithms.
\subsection{Effect of \ac{PMD} on entangled photon pairs}
\label{sec:system}
Consider a quantum network illustrated in Fig.~\ref{fig_PMD}, in which a photon source is connected to two network nodes, i.e., Alice and Bob, via  \ac{PMD}-degraded optical fibers.
The \ac{PMD} effect in the two fibers is parametrized by the \ac{PSP} basis $\{|s_i\rangle, |s'_i\rangle\}$ and differential group delay $\tau_i$, where 
$i\in\{\mathrm{A}, \mathrm{B}\}$ is the node index.
The source prepares a pair of  polarization-entangled photon pairs in maximally entangled states, and sends one photon to each network node.
Due to the effect of \ac{PMD}, the density matrix of the photon pair after passing through fibers is given by  \eqref{eqn:Gen_Density}.
The density matrix $\V{\rho}$ is written in the basis of $\vert s_{\mathrm{A}} s_{\mathrm{B}} \rangle$, $\vert s_{\mathrm{A}} s'_{\mathrm{B}} \rangle$, $\vert s'_{\mathrm{A}} s_{\mathrm{B}} \rangle$, and $\vert s'_{\mathrm{A}} s'_{\mathrm{B}} \rangle$.
Please refer to Appendix~\ref{sec:PMD} for the detailed derivation and the definition of the parameters in  \eqref{eqn:Gen_Density}.
Denote the element in the $p$-th row and $q$-th column of $\V{\rho}$ as $\rho_{pq}$.



As illustrated in Fig.~\ref{fig_PMD}B and \eqref{eqn:Gen_PMD}, with generic \ac{PSP}, 
the \ac{PMD} effect in the two arms leads to four possible coincident arrival times for the two photons, i.e., slow-slow ($\vert s_{\mathrm{A}} s_{\mathrm{B}} \rangle$), slow-fast ($\vert s_{\mathrm{A}} s'_{\mathrm{B}} \rangle$), fast-slow ($\vert s'_{\mathrm{A}} s_{\mathrm{B}} \rangle$), and fast-fast ($\vert s'_{\mathrm{A}} s'_{\mathrm{B}} \rangle$).
This results in a relatively complicated density matrix.
As illustrated in Fig.~\ref{fig_PMD}C, to simplify the density matrix, one could align the \ac{PSP} basis with the photon polarization basis, so that there are only two possible coincident arrival times, i.e., slow-slow and fast-fast.
The physical realization of this operation requires a measurement of the \ac{PSP} for a given fiber and the ability to perform local rotation on the photons before passing through the fiber.
As Appendix~\ref{pf_align} shows, local rotation on one of the photons is sufficient to achieve
the alignment of the \ac{PSP} basis with the photon polarization basis.
Existing studies suggest realignment of these states would be rare, as the \ac{PSP} in installed fiber optics can remain unchanged for as long as months \cite{BroFriBorTur:JLT06}.
In fact, the operation of aligning \ac{PSP} has also been adopted in the algorithm design for \ac{PMD} compensation \cite{ShtAntBro:OE11} to exploit the advantage of having a \ac{DFS} \cite{AntShtBro:PRL11}.

When the \ac{PSP} basis is aligned with the polarization basis, $\eta_{1}=1$ and $\eta_{2}=0$.
Hence, the density matrix \eqref{eqn:Gen_Density} is simplified to a matrix  with four non-zero elements, which are given by
\begin{align*}
\rho_{11}&=\rho_{44}=\frac{1}{2},\\
\rho_{41}&=\rho_{14}^{\dag}=\frac{1}{2}e^{\imath\alpha}R(\tau_{\mathrm{A}},\tau_{\mathrm{B}})
\end{align*}
which can be rewritten as
\begin{align}
\V{\rho}=&\frac{1}{2}\big(|s_{\mathrm{A}}s_{\mathrm{B}}\rangle\langle s_{\mathrm{A}}s_{\mathrm{B}}| + e^{-\imath\alpha}R^\dag(\tau_{\mathrm{A}},\tau_{\mathrm{B}})|s_{\mathrm{A}}s_{\mathrm{B}}\rangle\langle s'_{\mathrm{A}}s'_{\mathrm{B}}|\nonumber\\ 
&+ e^{\imath\alpha}R(\tau_{\mathrm{A}},\tau_{\mathrm{B}})|s'_{\mathrm{A}}s'_{\mathrm{B}}\rangle\langle s_{\mathrm{A}}s_{\mathrm{B}}| + |s'_{\mathrm{A}}s'_{\mathrm{B}}\rangle\langle s'_{\mathrm{A}}s'_{\mathrm{B}}|\big).\label{eqn:intialDM_2}
\end{align}

\subsection{Problem formulation}
The network nodes Alice and Bob adopt a recurrence \ac{QED} algorithm to remove the effect of \ac{PMD}.
They operate separately on every two qubit pairs, trying to improve the quality of entanglement in one pair at the expense of the other pair.
This distillation operation $\mathcal{D}$ can be formulated as follows.
Denote the density matrix of a kept qubit pair after $k$-th round of distillation as $\V{\rho}_k$, with $\V{\rho}_0=\V{\rho}$.
Then before the $k$-th round of distillation, the joint density matrix of two kept qubit pairs is given by
\begin{align*}\V{\rho}^{\mathrm{J}}_{k-1}={\V{\rho}_{k-1}}\otimes{\V{\rho}_{k-1}}
\end{align*}

Without loss of generality, assume that the network nodes try to keep the first qubit pair, i.e., the first and second qubits in the system.
Then the density matrix of the first qubit pair after the distillation operation is given by the partial trace over the third and fourth qubits normalized by the overall trace of the density matrix, i.e.,
\begin{align}
\V{\rho}_{k}=\frac{\mathrm{tr}_{3,4}\{\mathcal{D}\{\V{\rho}^{\mathrm{J}}_{k-1}\}\}}
{\mathrm{tr}\{\mathcal{D}\{\V{\rho}^{\mathrm{J}}_{k-1}\}\}}\label{eqn:rho_1}
\end{align}
where the distillation operation $\mathcal{D}$ must be in the category of \ac{LOCC}, and the probability of successfully keeping the first qubit pair is given by
\begin{align}
P_{k}=\mathrm{tr}\{\mathcal{D}\{\V{\rho}^{\mathrm{J}}_{k-1}\}\}.\label{eqn:Pk}
\end{align}
Denote the fidelity of the kept qubit pairs after the $k$-th round of distillation \ac{w.r.t.} to the targeted state as
\begin{align}
F_k = \langle \Phi^+| \V{\rho}_k |\Phi^+\rangle\label{eqn:Fk}
\end{align}
where $|\Phi^+\rangle=\frac{1}{\sqrt{2}}(|h_{\mathrm{A}}h_{\mathrm{B}}\rangle+|v_{\mathrm{A}}v_{\mathrm{B}}\rangle)$.
For notation convenience, denote the mapping between the input density matrix $\V{\rho}_{k-1}$ and the fidelity of the kept qubit pair $F_k$ as $F_{\mathcal{D}}$, i.e., 
\begin{align*}
F_k = F_{\mathcal{D}}(\V{\rho}_{k-1})
\end{align*}
and denote the mapping between the input density matrix $\V{\rho}_{k-1}$ and the success probability $P_k$ as 
$P_{\mathcal{D}}$, i.e., 
\begin{align*}
P_k = P_{\mathcal{D}}(\V{\rho}_{k-1}).
\end{align*}
Note that both mappings depend on the distillation operation $\mathcal{D}$.

The objective of recurrence \ac{QED} algorithms is to generate qubit pairs with sufficiently high fidelity, i.e.,
\begin{align}
F_K \ge 1- \epsilon\label{eqn:objective}
\end{align}
for some natural number $K$ and small $\epsilon>0$.  
With this recurrence \ac{QED} algorithm, the yield of the algorithm after $K$ rounds of distillation is given by
\begin{align}
Y_K = \prod_{k=1}^{K}\frac{P_k}{2}\label{eqn:yield}
\end{align}
It can be seen from \eqref{eqn:yield} that the yield of the algorithm drops by at least half with one more round of distillation.
Hence, to improve the yield of the \ac{QED} algorithm, a primary task is to minimize the required rounds of distillation, i.e., 
maximize $F_k$.
Meanwhile, the success probability $P_k$ also affects $Y_K$.
Hence, a secondary task is to maximize $P_k$ conditional on $F_k$ being maximized.
The problems of fulfilling these two tasks are formulated as follows.


In a certain round of distillation, given the input density matrix $\V{\rho}$, 
we will maximize the fidelity of the kept qubit pair $F_{\mathcal{D}}(\V{\rho})$ \ac{w.r.t.} the distillation operation ${\mathcal{D}}$.
This problem can be formulated as 
\begin{align*}
\Ps_{F}:&\quad \max_{\mathcal{D}}F_{\mathcal{D}}(\V{\rho})
\end{align*}
Denote the optimal fidelity as $F^*(\V{\rho})$.
We will maximize the success probability of the distillation operation $P_{\mathcal{D}}(\V{\rho})$ \ac{w.r.t.} the distillation operation ${\mathcal{D}}$ conditional on the optimal fidelity being achieved.
This problem can be formulated as:
\begin{align*}
\Ps_{P}:&\quad\max_{\mathcal{D}}P_{\mathcal{D}}(\V{\rho})\\
\mbox{s.t.}&\quad F_{\mathcal{D}}(\V{\rho})= F^*(\V{\rho}).
\end{align*}

\section{Efficient \ac{QED} for \ac{PMD} channels}
\label{sec:upper_alg}
This section will first characterize the optimal performance of problems $\Ps_{F}$ and $\Ps_{P}$, and then give an algorithm which achieves the optimal performance in every round of distillation.
For conciseness, in the following, both $|h_{\mathrm{A}}\rangle$ and $|h_{\mathrm{B}}\rangle$ are denoted as $|0\rangle$, and both $|v_{\mathrm{A}}\rangle$ and $|v_{\mathrm{B}}\rangle$ are denoted as $|1\rangle$. The network node index can be omitted without causing confusion because only local operations are involved in the distillation process.

\subsection{Characterization of performance upper bounds}
\label{sec:upper}
This subsection considers a set of density matrices that includes the density matrices given in \eqref{eqn:intialDM_2}, and characterizes the corresponding optimal performance of problems $\Ps_{F}$ and $\Ps_{P}$.
Specifically, the set of density matrices is defined as 
\begin{align*}
\mathcal{S}=\{\V{\rho} \mbox{ that satisfies } \eqref{eqn:generalrho}\}
\end{align*}
where
\begin{align}
\V{\rho}&=\frac{1}{2}\big(|ab\rangle\langle ab| + e^{-\imath\alpha}R^{\dag}|ab\rangle\langle a'b'| \nonumber\\
&\hspace{4.3mm}+ 
e^{\imath\alpha}R|a'b'\rangle\langle ab| + |a'b'\rangle\langle a'b'|\big).\label{eqn:generalrho}
\end{align}
in which 
\begin{align*}
\langle x | x' \rangle &= 0, \quad x\in\{a,b\},\\
\alpha &\in [0,2\pi),\quad \mbox{and}\\
|R|&\in[0,1].
\end{align*}

First simplify the initial density matrix $\V{\rho}$ in \eqref{eqn:generalrho}. 
By performing spectrum decomposition, it can be obtained that
\begin{align}
\V{\rho}&= F|\phi_1\rangle\langle\phi_1| +(1-F) |\phi_2\rangle\langle\phi_2| \label{eqn:generalrho_2}
\end{align}
where
\begin{align*}
F &= \frac{1}{2} (1+ |R|)\\
|\phi_1\rangle &= \frac{1}{\sqrt{2}}(|ab\rangle+e^{\imath\theta}|a'b'\rangle)\\
|\phi_2\rangle &= \frac{1}{\sqrt{2}}(|ab\rangle-e^{\imath\theta}|a'b'\rangle)\\
\theta &= \alpha + \mathrm{Phase}\{R\}
\end{align*}

The following theorem characterizes the optimal fidelity that can be achieved when input density matrix $\V\rho \in\mathcal{S}$.

\begin{Thm}[Optimal fidelity] \thlabel{thm:UBF} When $\V\rho \in\mathcal{S}$, the optimal performance of $\Ps_{F}$  is given by
\begin{align}
F^*(\V{\rho}) = \frac{F^2}{F^2+(1-F)^2}.
\label{eqn:Fidelity_Upper}
\end{align}
\end{Thm}

\begin{proof} 
Please refer to Appendix~\ref{pf_thm:UBF} for the proof.
\end{proof}

The next theorem characterizes the upper bound of the success probability conditional on the optimal fidelity having been achieved.
\begin{Thm}[Optimal probability of success] \thlabel{thm:UBP} 
When \\$\V\rho \in\mathcal{S}$ with $|R|>0$, the optimal performance of $\Ps_{P}$  is given by
\begin{align}
P^*(\V{\rho}) = F^2+(1-F)^2.
\label{eqn:Prob_Upper}
\end{align}
\end{Thm}

\begin{proof}
Please refer to Appendix~\ref{pf_thm:UBP} for the proof.
\end{proof}

\subsection{Algorithm design}
\label{sec:alg}
The two theorems in the previous subsection characterize the optimal fidelity and the corresponding optimal success probability of distillation operations on two pairs of qubits.
In this subsection, guided by the insights obtained from the proofs of \thref{thm:UBF} and \thref{thm:UBP},
the following recurrence \ac{QED} algorithm is designed to achieve the optimal fidelity and the corresponding optimal success probability in every round of distillation.

\begin{Algno}[Efficient \ac{QED} for \ac{PMD} channel]
\label{alg:dis}
\begin{itemize}
\item[]
\item{\bf Local state preparation:} For each qubit pair, the network nodes transform the density matrix to $\check{\rho}$ using local unitary operators $\M{U}_\mathrm{A}$ and $\M{U}_\mathrm{B}$ defined in \eqref{eqn:UAB}.

\item{\bf First round distillation:} The nodes take two of the kept qubit pairs, perform the following operations, and repeat these operations on all kept qubit pairs.

(i) Each node locally performs CNOT operation, i.e., $\M{U}=|00\rangle\langle 00|+|01\rangle\langle 01|+|10\rangle\langle 11|+|11 \rangle\langle 10|$ on the two qubits at hand.

(ii) Each node measures the target bit (i.e., the qubit in the second pair) using operators $|0\rangle\langle 0|$, $|1\rangle\langle 1|$, and transmits the measurement result to the other node via classical communication.

(iii) If their measurement results do not agree, the nodes discard the source qubit pair (i.e., the first pair). Otherwise, the nodes keep the source qubit pair.

\item{\bf Following rounds:} Network nodes perform the same operations as in the first round, until the fidelity of the kept qubit pairs exceeds the required threshold.~\QEDB
\end{itemize}
\end{Algno}

In the following, we will first characterize the performance of the proposed algorithm in \thref{thm:dis_perf}, then explain the implications of this theorem in two remarks.

\begin{Thm}[Performance of the proposed algorithm]\thlabel{thm:dis_perf}
In the $k$-th round of distillation, the source qubit pair is kept with fidelity
\begin{align}
F_{k}=\frac{F_{k-1}^2}{F_{k-1}^2+(1-F_{k-1})^2}\label{eqn:F_update}
\end{align}
probability 
\begin{align}P_k=F_{k-1}^2+(1-F_{k-1})^2 \label{eqn:P_update}\end{align} 
and density matrix
\begin{align}
\V{\rho}^{(k)}=F_k|\Phi^+\rangle\langle\Phi^+| +(1-F_{k})|\Psi^+\rangle\langle\Psi^+|.\label{eqn:rhoform}
\end{align}

\end{Thm}
\begin{proof}
The proof is given in Appendix~\ref{pf_thm:dis_perf}.
\end{proof}

\begin{Remark}[Optimality of the proposed algorithm] \label{remark:converge}
In Theorem~\ref{thm:dis_perf}, \eqref{eqn:rhoform} shows that the proposed algorithm always keeps the density matrix of qubit pairs in set $\mathcal{S}$, 
which means that the results in \thref{thm:UBF} and \thref{thm:UBP} apply to every round of distillation.
Therefore, by comparing \eqref{eqn:Fidelity_Upper}, \eqref{eqn:Prob_Upper} with \eqref{eqn:F_update}, \eqref{eqn:P_update},
one can see that the proposed algorithm achieves the optimal fidelity and the corresponding optimal success probability in every round of distillation.
This feature enables the proposed algorithm to achieve high efficiency.
~\hfill~\QEDB
\end{Remark}

\begin{Remark}[Convergence speed of fidelity]
In terms of the convergence speed of fidelity \ac{w.r.t.} the rounds of distillation in recurrence \ac{QED} algorithms, the only existing theoretical result was given in \cite{BenBraPopSchSmoWoo:96}, which shows that the fidelity of kept qubit pairs in consecutive rounds is
\begin{align}
{F}_k=\frac{F_{k-1}^2+\frac{1}{9}(1-F_{k-1})^2}{F_{k-1}^2+\frac{2}{3}F_{k-1}(1-F_{k-1})+\frac{5}{9}(1-F_{k-1})^2}\,.\label{eqn:speed-old}
\end{align}
In this case, when $F_0>\frac{1}{2}$, it can be obtained that
\begin{align}
\lim_{k\rightarrow \infty}\frac{1-F_k}{1-F_{k-1}}=\frac{2}{3} \,. \label{eqn:speed1}
\end{align}
For the proposed algorithms, it can be shown from \eqref{eqn:F_update} that
when  $F_0>\frac{1}{2}$
\begin{align}
\lim_{k\rightarrow \infty}\frac{1-F_k}{1-F_{k-1}}=0 \,, \qquad
 \lim_{k\rightarrow \infty}\frac{1-F_k}{(1-F_{k-1})^2}=1 \,.\label{eqn:speed2}
\end{align}
Equation
\eqref{eqn:speed1} shows that with the algorithm proposed in \cite{BenBraPopSchSmoWoo:96}, the fidelity of the qubit pairs converges to $1$ linearly at rate $\frac{2}{3}$, whereas \eqref{eqn:speed2} shows that with the proposed algorithms, the fidelity converges to $1$ quadratically.
Hence, the convergence speed of our algorithm is quadratic in number of iteration rounds, which is a significant improvement over the linear convergence achieved by the recurrence \ac{QED} algorithm proposed in \cite{BenBraPopSchSmoWoo:96}.

Hence, the convergence speed of the proposed algorithm is improved from linear to quadratic.~\hfill~\QEDB
\end{Remark}

\section{Numerical Results}
\label{sec:numerical}
We will now demonstrate the dependence of the proposed recurrence distillation \ac{QED} algorithm on the parameters of the \ac{PMD} channel by numerically calculating the yield and output fidelity for different channel configurations. 
We compare the yield of our algorithm with that obtained by an existing recurrence \ac{QED} algorithm \cite{BenBraPopSchSmoWoo:96}. 
As an additional benchmark, an upper bound of yield derived from distillable entanglement \cite{Rai:99,Rai:01} is also calculated and plotted.
While the achievability of this bound remains unknown, it is arguably the best known upper bound on the yield of any QED algorithms \cite{WanDua:17}. 
We find that our algorithm has a significant performance advantage in parameter regimes where partial \ac{PMD} compensation occurs \cite{ShtAntBro:OE11,AntShtBro:PRL11},
and achieves a yield close to the theoretical upper bound despite its simple recurrent distillation operations that involve only two qubit pairs.
Additionally, we have performed tests to examine how robust the proposed algorithm is to basis alignment errors.

To perform numerical tests, one needs to first specify the optical properties of the entanglement source in order to determine the form for $R(\tau_{\mathrm{A}},\tau_{\mathrm{B}})$ (see Appendix~\ref{sec:PMD} for more details).
We assume that the frequency content of the pulsed pump laser and frequency response of the filters are Gaussian.
Under this assumption, the form of $R(\tau_{\mathrm{A}},\tau_{\mathrm{B}})$ is given by \cite{BroGeoAntSht:OL11,ShtAntBro:OE11}
\begin{align*}
R(\tau_{\mathrm{A}},\tau_{\mathrm{B}})&= \kappa \int \int \text{d}\omega_{\mathrm{A}} \text{d}\omega_{\mathrm{B}} \vert H_{\mathrm{A}}(\omega_{\mathrm{A}})\vert^{2} \vert H_{\mathrm{B}}(\omega_{\mathrm{B}})\vert^{2} \\
&\hspace{14.3mm}\left| \tilde{E}_{p}\left(\omega_{\mathrm{A}}+\omega_{\mathrm{B}}\right)\right|^{2}e^{\imath (\tau_{\mathrm{A}} \omega_{\mathrm{A}}+\tau_{\mathrm{B}}\omega_{\mathrm{B}})}
\end{align*}
where $\tilde{E}_{p}(\omega)\propto e^{-\omega/4 B_{p}^{2}}$, $H_{i}(\omega) \propto e^{-(\omega \pm \Delta\Omega)^{2}/4B_{i}^{2}}$, $i\in\{\mathrm{A},\mathrm{B}\}$, with the $B_{i}$ terms representing the root mean square bandwidth of each filter. The central frequency of the pump is set to zero and Alice and Bob's filters are each offset from it by $\pm \Delta \Omega$.
The integral results in:
\begin{align*}
R(\tau_{\mathrm{A}},\tau_{\mathrm{B}})=e^{-\frac{B_{\mathrm{A}}^2 B_{\mathrm{B}}^2 (\tau_{\mathrm{A}}-\tau_{\mathrm{B}})^2+B_{\mathrm{A}}^2B_{p}^2 \tau_{\mathrm{A}}^2+B_{\mathrm{B}}^2 B_{p}^2 \tau_{\mathrm{B}}^2}{2 \left(B_{\mathrm{A}}^2+B_{\mathrm{B}}^2+B_{p}^2\right)}}e^{- i \Delta \Omega (\tau_{\mathrm{A}}-\tau_{\mathrm{B}})}
\end{align*}

\begin{figure}[t] \centering
\psfrag{0}[Br][Br][0.7]{0\hspace{0.3mm}}
\psfrag{0.1}[Br][Br][0.7]{0.1\hspace{0.3mm}}
\psfrag{0.2}[Br][Br][0.7]{0.2\hspace{0.3mm}}
\psfrag{0.3}[Br][Br][0.7]{0.3\hspace{0.3mm}}
\psfrag{0.4}[Br][Br][0.7]{0.4\hspace{0.3mm}}
\psfrag{0.5}[Br][Br][0.7]{0.5\hspace{0.3mm}}
\psfrag{0.6}[Br][Br][0.7]{0.6\hspace{0.3mm}}
\psfrag{0.7}[Br][Br][0.7]{0.7\hspace{0.3mm}}
\psfrag{0.8}[Br][Br][0.7]{0.8\hspace{0.3mm}}
\psfrag{0.9}[Br][Br][0.7]{0.9\hspace{0.3mm}}
\psfrag{1}[Br][Br][0.7]{1\hspace{0.3mm}}
\psfrag{0.0}[tt][tt][0.7]{0}
\psfrag{0.50}[tt][tt][0.7]{0.5}
\psfrag{1.0}[tt][tt][0.7]{1}
\psfrag{1.5}[tt][tt][0.7]{1.5}
\psfrag{2}[tt][tt][0.7]{2}
\psfrag{2.5}[tt][tt][0.7]{2.5}
\psfrag{3}[tt][tt][0.7]{3}
\psfrag{Yield}[tc][tc][0.8]{Yield}
\psfrag{tauB/tauA}[tc][tc][0.8]{$\tau_{\mathrm{B}}/\tau_{\mathrm{A}}$}
\psfrag{BBPSSW}[cl][cl][0.7]{\hspace{-0.mm}BBPSSW algorithm}
\psfrag{Proposed                    1}[cl][cl][0.7]{\hspace{-0.mm}Proposed algorithm}
\psfrag{theory}[cl][cl][0.7]{\hspace{-0.mm}Upper bound}
\psfrag{Gain:450}[cl][cl][0.7]{Gain: 450\%}
\psfrag{Gap:39}[cl][cl][0.7]{Gap: 36\%}
\includegraphics[scale=0.65]{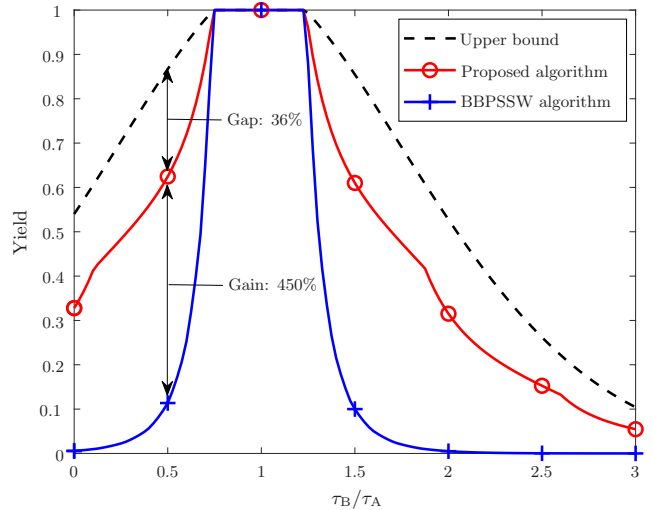}
\caption {Comparison of the yield as a function of $\tau_{\mathrm{B}}/\tau_{\mathrm{A}}$ for the proposed algorithm and the benchmarks, i.e., the upper bound \cite{Rai:99} and the BBPSSW algorithm \cite{BenBraPopSchSmoWoo:96}. In this plot, $B_{p}=0.1$,$B_{\mathrm{A}}=B_{\mathrm{B}}=1$, $\tau_{\mathrm{A}}=1$.}
\label{fig_Distillation_Efficiency_PMD_Case1}
\end{figure}
\begin{figure}[t] \centering
\psfrag{0}[Br][Br][0.7]{0\hspace{0.3mm}}
\psfrag{0.1}[Br][Br][0.7]{0.1\hspace{0mm}}
\psfrag{0.2}[Br][Br][0.7]{0.2\hspace{0.3mm}}
\psfrag{0.3}[Br][Br][0.7]{0.3\hspace{0.3mm}}
\psfrag{0.4}[Br][Br][0.7]{0.4\hspace{0.3mm}}
\psfrag{0.5}[Br][Br][0.7]{0.5\hspace{0.3mm}}
\psfrag{0.0}[tt][tt][0.7]{0}
\psfrag{0.50}[tt][tt][0.7]{0.5}
\psfrag{1.0}[tt][tt][0.7]{1}
\psfrag{1.5}[tt][tt][0.7]{1.5}
\psfrag{2}[tt][tt][0.7]{2}
\psfrag{2.5}[tt][tt][0.7]{2.5}
\psfrag{3}[tt][tt][0.7]{3}
\psfrag{tauB/tauA}[tc][tc][0.8]{$\tau_{\mathrm{B}}/\tau_{\mathrm{A}}$}
\psfrag{Yield}[tc][tc][0.8]{Yield}
\psfrag{BBPSSW}[cl][cl][0.7]{\hspace{-0.mm}BBPSSW algorithm}
\psfrag{Proposed                    1}[cl][cl][0.7]{\hspace{-0.mm}Proposed algorithm}
\psfrag{theory}[cl][cl][0.7]{\hspace{-0.mm}Upper bound}
\psfrag{Gain:5660}[cl][cl][0.7]{Gain: 5660\%}
\psfrag{Gap:53}[cl][cl][0.7]{Gap: 53\%}
\includegraphics[scale=0.65]{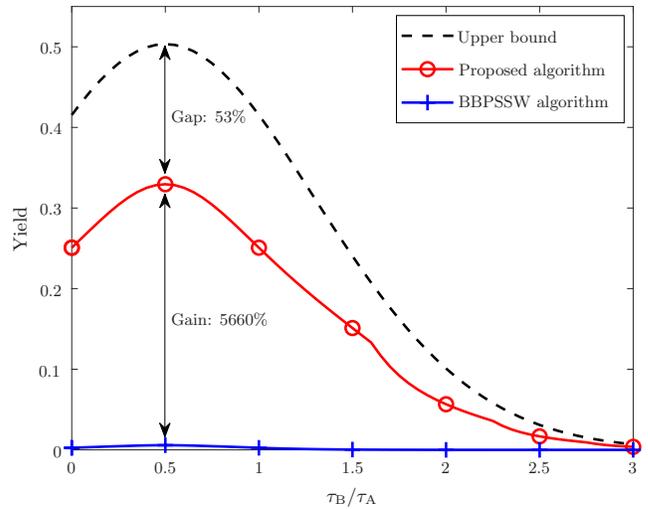}
\caption {Comparison of the yield as a function of $\tau_{\mathrm{B}}/\tau_{\mathrm{A}}$ for the proposed algorithm and the benchmarks, i.e., the upper bound \cite{Rai:99} and the BBPSSW algorithm \cite{BenBraPopSchSmoWoo:96}. In this plot, $B_{p}=1$,$B_{\mathrm{A}}=B_{\mathrm{B}}=1$, $\tau_{\mathrm{A}}=1$.}
\label{fig_Distillation_Efficiency_PMD_Case2}
\end{figure}

In the numerical tests, the targeted fidelity is set to be 0.99. 
The round of distillation $K$ is set to be the minimum round that achieves the targeted fidelity, and the yield of the algorithm is calculated according to \eqref{eqn:yield}.
We assume that the photon bandwidths $B_{\mathrm{A}}$ and $B_{\mathrm{B}}$ are equal, and we set $\tau_{\mathrm{A}}B_{\mathrm{A}}=1$ while varying the \ac{DGD} on photon $B$, given by $\tau_{B}$, the pump laser bandwidth $B_{p}$, and $\eta$, which specifies the alignment between the qubit and \ac{PSP} basis.

Figs.~\ref{fig_Distillation_Efficiency_PMD_Case1} and~\ref{fig_Distillation_Efficiency_PMD_Case2} plot the yield as a function of the ratio of the magnitudes of the \ac{DGD} in each optical path for two different pulse pump bandwidths.
Fig.~\ref{fig_Distillation_Efficiency_PMD_Case1} plots the case where the pump bandwidth is given by $B_{p}=0.1/\tau_{\mathrm{A}}$, which corresponds to a relatively long pump duration as compared to the \ac{DGD}.
Alternatively, Fig.~\ref{fig_Distillation_Efficiency_PMD_Case2} plots a case where a pump bandwidth is on the order of the \ac{DGD}, given by $B_{p}=1/\tau_{\mathrm{A}}$.

In Fig.~\ref{fig_Distillation_Efficiency_PMD_Case1} we see that both algorithms achieve a yield of unity for a finite region of $\tau_{\mathrm A}/\tau_{\mathrm B}$ centered around the \ac{DFS} at $\tau_{\mathrm A}=\tau_{\mathrm B}$ \cite{ShtAntBro:OE11,AntShtBro:PRL11}.
For regions of partial or no compensation, the regions outside of unit yield in Fig.~\ref{fig_Distillation_Efficiency_PMD_Case1} and all of Fig.~\ref{fig_Distillation_Efficiency_PMD_Case2}, the proposed algorithm achieves a fidelity that is significantly higher than the baseline algorithm from \cite{BenBraPopSchSmoWoo:96} and is reasonably close to the best known upper bound.
For instance, when $\tau_{\mathrm B}/\tau_{\mathrm A}=0.5$, the proposed algorithm increases the yield from 450\% to 5660\% compared to the baseline algorithm and is 36\% to 53\% away from the upper bound.
Given that the proposed algorithm adopts simple recurrent distillation operations that involve only two qubit pairs, it achieves a desirable balance between efficiency and implementability.
We also note that the peak of the yield for both algorithms in Fig.~\ref{fig_Distillation_Efficiency_PMD_Case2} is shifted away from $\tau_{\mathrm{A}}=\tau_{\mathrm{B}}$, as opposed to the peak being centered around this point in Fig.~\ref{fig_Distillation_Efficiency_PMD_Case1}.  
This observation is consistent with those of \cite{ShtAntBro:OE11} on \ac{PMD} compensation, which emphasizes the fact that our algorithm attempts to make use of nonlocal \ac{PMD} compensation to whatever extent is possible.

\begin{figure}[t] \centering
\psfrag{0}[Br][Br][0.7]{0\hspace{0.3mm}}
\psfrag{0.1}[Br][Br][0.7]{0.1\hspace{-0mm}}
\psfrag{0.2}[Br][Br][0.7]{0.2\hspace{0.3mm}}
\psfrag{0.3}[Br][Br][0.7]{0.3\hspace{0.3mm}}
\psfrag{0.4}[Br][Br][0.7]{0.4\hspace{0.3mm}}
\psfrag{0.5}[Br][Br][0.7]{0.5\hspace{0.3mm}}
\psfrag{0.6}[Br][Br][0.7]{0.6\hspace{0.3mm}}
\psfrag{0.7}[Br][Br][0.7]{0.7\hspace{0.3mm}}
\psfrag{0.8}[Br][Br][0.7]{0.8\hspace{0.3mm}}
\psfrag{0.9}[Br][Br][0.7]{0.9\hspace{0.3mm}}
\psfrag{1}[Br][Br][0.7]{1\hspace{-0mm}}
\psfrag{0.0}[tt][tt][0.7]{0}
\psfrag{0.20}[tt][tt][0.7]{0.2}
\psfrag{0.40}[tt][tt][0.7]{0.4}
\psfrag{0.60}[tt][tt][0.7]{0.6}
\psfrag{0.80}[tt][tt][0.7]{0.8}
\psfrag{1.0}[tt][tt][0.7]{1}
\psfrag{1.2}[tt][tt][0.7]{1.2}
\psfrag{1.4}[tt][tt][0.7]{1.4}
\psfrag{1.6}[tt][tt][0.7]{1.6}
\psfrag{1.8}[tt][tt][0.7]{1.8}
\psfrag{2}[tt][tt][0.7]{2}
\psfrag{Bp}[tc][tc][0.8]{$B_p$}
\psfrag{Yield}[tc][tc][0.8]{Yield}
\psfrag{taub=      0.1}[cl][cl][0.8]{\hspace{0.5mm}$\tau_{\mathrm{B}}=0.1$}
\psfrag{taub=0.5}[cl][cl][0.8]{\hspace{0.5mm}$\tau_{\mathrm{B}}=0.5$}
\psfrag{taub=0.9}[cl][cl][0.8]{\hspace{0.5mm}$\tau_{\mathrm{B}}=0.9$}
\psfrag{taub=1.3}[cl][cl][0.8]{\hspace{0.5mm}$\tau_{\mathrm{B}}=1.3$}
\includegraphics[scale=0.66]{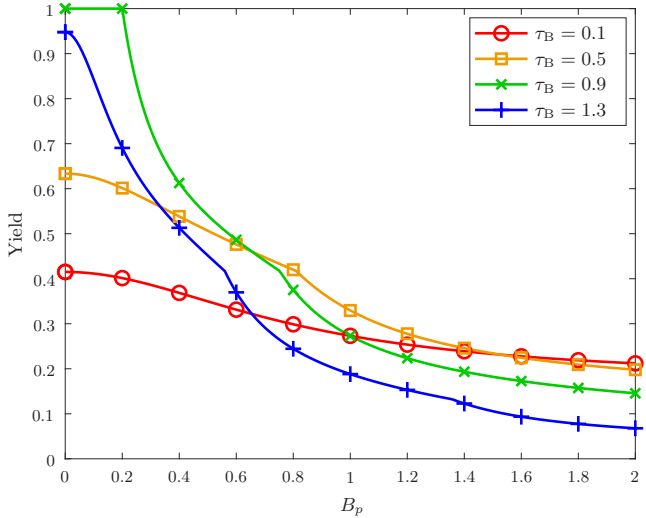}
\caption {The efficiency of the proposed algorithm as a function of the bandwidth of the source laser pump. In this figure, $B_{\mathrm{A}}=B_{\mathrm{B}}=1$, $\tau_{\mathrm{A}}=1$.}
\label{fig_PMD_Bp}
\end{figure}

To further demonstrate the impact of pump bandwidth on the performance of the proposed algorithm, the yield as a function of $B_{p}$ is plotted in Fig.~\ref{fig_PMD_Bp} for several values of $\tau_{\mathrm{B}}$.
From the figure, it can be observed that the yield of the algorithm is a decreasing function of the pump bandwidth $B_{p}$. 
This is because the larger $B_{p}$ is, the more distinguishable are the photon pairs advanced and delayed  by \ac{PMD}.
For analogous reasons, we see that when $B_p$ is large, the yield of the algorithm is a decreasing function of $\tau_{\mathrm{B}}$.
However, when $B_p$ is small, the yield of the algorithm is highest when the values of $\tau_{\mathrm{A}}$, $\tau_{\mathrm{B}}$ are similar, illustrating the benefits of the \ac{DFS} created by \ac{PMD} compensation.

\begin{figure}[t] \centering
\psfrag{0.1}[Br][Br][0.7]{0.1\hspace{0.3mm}}
\psfrag{0.2}[Br][Br][0.7]{0.2\hspace{0.3mm}}
\psfrag{0.3}[Br][Br][0.7]{0.3\hspace{0.3mm}}
\psfrag{0.4}[Br][Br][0.7]{0.4\hspace{0.3mm}}
\psfrag{0.5}[Br][Br][0.7]{0.5\hspace{0.3mm}}
\psfrag{0.6}[Br][Br][0.7]{0.6\hspace{0.3mm}}
\psfrag{0.7}[Br][Br][0.7]{0.7\hspace{0.3mm}}
\psfrag{0.75}[Br][Br][0.7]{0.75\hspace{0.3mm}}
\psfrag{0.8}[Br][Br][0.7]{0.8\hspace{0.3mm}}
\psfrag{0.85}[Br][Br][0.7]{0.85\hspace{0.3mm}}
\psfrag{0.9}[Br][Br][0.7]{0.9\hspace{0.3mm}}
\psfrag{0.95}[Br][Br][0.7]{0.95\hspace{0.3mm}}
\psfrag{1}[Br][Br][0.7]{1\hspace{0.3mm}}
\psfrag{0}[tt][tt][0.7]{0}
\psfrag{5}[tt][tt][0.7]{5}
\psfrag{10}[tt][tt][0.7]{10}
\psfrag{15}[tt][tt][0.7]{15}
\psfrag{20}[tt][tt][0.7]{20}
\psfrag{25}[tt][tt][0.7]{25}
\psfrag{30}[tt][tt][0.7]{30}
\psfrag{Misalignment angle (degree)}[tc][tc][0.8]{Misalignment angle (degree)}
\psfrag{Yield}[tc][tc][0.8]{Yield}
\psfrag{Output Fidelity}[tb][tb][0.8]{Output Fidelity}
\psfrag{tau=  0.2}[cl][cl][0.8]{\hspace{-0.3mm}$\tau=0.2$}
\psfrag{tau=  0.5}[cl][cl][0.8]{\hspace{-0.3mm}$\tau=0.5$}
\psfrag{tau=  1}[cl][cl][0.8]{\hspace{-0.3mm}$\tau=1$}
\includegraphics[scale=0.66]{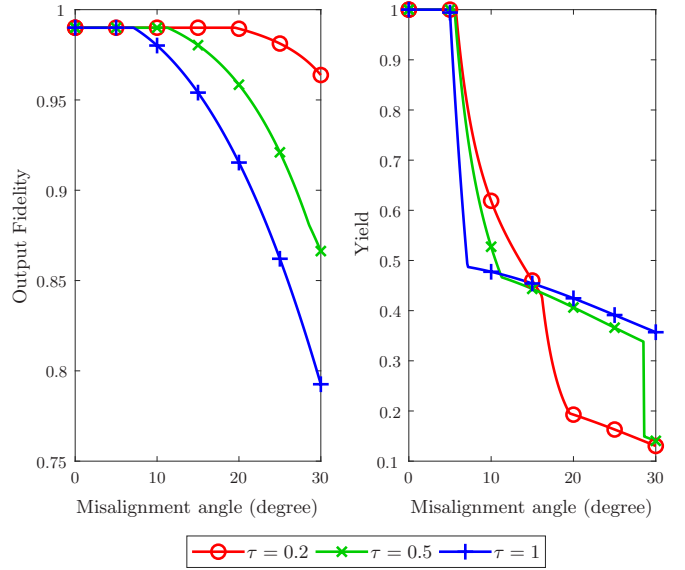}
\caption {The output fidelity and the efficiency of the proposed algorithm as a function of the misalignment angle $\theta$. $\eta_1 = \arcsin(\frac{\theta\pi}{180})$. In this figure, $B_{\mathrm{A}}=B_{\mathrm{B}}=1$, $B_{p}=0.1$, $\tau_{\mathrm{A}}=\tau_{\mathrm{B}}=\tau$. 
The output fidelity is the maximum achievable by the algorithm, up to a fidelity of $0.99$.
}
\label{fig_PMD_angle}
\end{figure}

Finally, the performance of the proposed algorithm is evaluated in the presence of basis alignment errors.
Until now is has been assumed that it is possible to locally rotate the polarization basis such that they perfectly align with the \ac{PSP} basis of the fiber.
As mentioned in Section~\ref{sec:system}, such an alignment is not expected to be performed frequently, as the \ac{PSP} of installed fiber optics has been shown to remain unchanged on the timescale of months \cite{BroFriBorTur:JLT06}.
However, any realistic implementation will have to deal with errors in the initial alignment process and the eventual drift of the \ac{PSP} with time.
To help us quantify the effects of implementation error on the performance of the proposed algorithm, we define the misalignment angle between the polarization and \ac{PSP} basis as $\theta$, where $\eta_{1}=\arcsin(\frac{\theta\pi}{180})$.
In Fig.~\ref{fig_PMD_angle} the output fidelity and the yield of the proposed algorithm are plotted as a function of misalignment angle $\theta$, for several values of $\tau$, where $\tau_{A}=\tau_{B}=\tau$.
The output fidelities shown in the plot are the maximum achievable fidelity with the proposed algorithm where the algorithm halts if it achieves a fidelity of $0.99$.
It can be seen that for all considered values of $\tau$, the algorithm can generate qubit pairs with required fidelity when the misalignment angle is no more than 5 degrees.
When the misalignment angle $\theta$ is greater than 5 degrees, the output fidelities are higher for smaller values of $\tau$, meaning that the robustness of the algorithm is inversely proportional to the magnitude of the \ac{DGD}.
Finally, it can be observed that the yield of the algorithm drops significantly when the misalignment angle $\theta$ is around 5 degrees.
This means that, even though the algorithm can still obtain photon pairs with high fidelity when $\theta>5$, it demands a significant increase in resources.
This result can be used to bound the precision of local unitary operations needed for an experimental implementation of this algorithm.

\section{Conclusion}
\label{sec:conclusion}
Recurrence \ac{QED} algorithms have good implementability and robustness, but improving their efficiency remains an interesting challenge.
This work adopts recurrence \ac{QED} algorithms to obtain high-quality entanglement from polarization-entangled photon pairs affected by \ac{PMD}-degraded channels.
For these photon pairs, we have characterized the optimal fidelity that can be achieved by recurrence \ac{QED} operations as well as the optimal success probability conditional on the optimal fidelity being achieved.
We then proposed a recurrence \ac{QED} algorithm that achieves both optimal fidelity and success probability in every round of distillation.
Analytical results show that the proposed algorithm improves the convergence speed of fidelity \ac{w.r.t.} the rounds of distillation from linear to quadratic.
Numerical tests show that the proposed algorithm significantly improves the efficiency of \ac{QED} in a wide range of operation regions, and achieves a yield close to the best known upper bound for any QED algorithms.

\appendix
\section{Analysis of the effect of \ac{PMD}}
\label{sec:PMD}

The effect of \ac{PMD} on a polarization-entangled photon pair depends on the way that the photons are generated, in particular, the type of nonlinear media and 
laser pump.
A rigorous treatment dealing with $\chi^{(3)}$ media and a \ac{CW} pump was given in \cite{AntShtBro:PRL11}, 
and the scenario with $\chi^{(2)}$ media and a pulsed pump was analyzed in \cite{ShtAntBro:OE11}.
Here we present an analytical treatment for $\chi^{(2)}$ media and a pulsed pump, and will also consider the limit where the frequency content of the pulse approaches a delta function, effectively becoming a \ac{CW} beam.

Consider a pair of photons which are entangled in two orthogonal polarizations as well as time.
These pairs can be created using parametric down conversion or fiber nonlinearities \cite{TakIno:PRA04,BurWei:PRL70}, and are notated as
\begin{equation}
\vert \psi \rangle = \vert f(t_{\mathrm{A}},t_{\mathrm{B}}) \rangle \otimes \frac{1}{\sqrt{2}}(\vert h_{\mathrm{A}} \rangle \vert h_{\mathrm{B}} \rangle + e^{\imath \alpha} \vert v_{\mathrm{A}} \rangle \vert v_{\mathrm{B}} \rangle ),
\label{eq:input_state}
\end{equation}
where $h_{i}$ and $v_{i}$ are orthogonal polarization basis states of photons $A$ and $B$. 
The term $\vert f(t_{\mathrm{A}},t_{\mathrm{B}})\rangle$ describes the time component of the state and is given by
\begin{equation}
\vert f(t_{\mathrm{A}},t_{\mathrm{B}}) \rangle = \int \int dt_{\mathrm{A}}dt_{\mathrm{B}}f(t_{\mathrm{A}},t_{\mathrm{B}})\vert t_{\mathrm{A}},t_{\mathrm{B}}\rangle.
\end{equation}
The function $\vert f(t_{\mathrm{A}},t_{\mathrm{B}})\vert^{2}$ is proportional to the probability that the two photons overlap in time, and $\int \text{d}t_{\mathrm{A}}\text{d}t_{\mathrm{B}}\vert f(t_{\mathrm{A}},t_{\mathrm{B}})\vert^{2}=1$.
Specifically, since the entanglement is generated via a $\chi^{(2)}$ media, this function can be written as 
\begin{align}f(t_{\mathrm{A}},t_{\mathrm{B}})=\int \text{d}t H_{\mathrm{A}}^{*}(t-t_{\mathrm{A}})H_{\mathrm{B}}^{*}(t-t_{\mathrm{B}})E_{p}(t)\label{eqn:Ep}
\end{align}
where $H_{i}^{*}(t)$ represents the inverse Fourier transform of the frequency filter $H_{i}(\omega)$ at node $i\in\{\mathrm{A},\mathrm{B}\}$ and $E_{p}(t)$ is the envelope of the pump signal.

The two types of laser pumps, \ac{CW} and pulsed, are characterized by the envelope of the pump signal $E_{p}(t)$ and its Fourier transform $\tilde{E}_{p}(\omega)$, which describes the frequency content of the input pulse.
Experimentally, pulsed pump lasers are convenient because they allow experiments to be broken into discrete detection time bins, and can result in wider bandwidth signal and idler photons, which enables multiple channels.
For \ac{CW} lasers, $\vert\tilde{E}_{p}(\omega)\vert^{2}$ approaches a delta function, which is a constant in the time domain.
In this case, $f(t_{\mathrm{A}},t_{\mathrm{B}})$ becomes a function of only the time difference, 
removing any absolute reference and hence simplifies analysis.

The effect of \ac{PMD} is to advance or delay photon arrival times, with the maximum and minimum alterations occurring for photons with polarizations equal to the PSP of the fiber \cite{AntShtBro:PRL11}.
Therefore, it is convenient to write the initial state in terms of the PSP basis $\{\vert  s_{i} \rangle, \vert  s'_{i} \rangle\}$, $i\in\{\mathrm{A}, \mathrm{B}\}$.
In this basis the initial state becomes
\begin{align}
\vert \psi_{PSP} \rangle &= \vert f(t_{\mathrm{A}},t_{\mathrm{B}}) \rangle \otimes \Big[\frac{\eta_{1}}{\sqrt{2}}\left(\vert s_{\mathrm{A}} \rangle \vert s_{\mathrm{B}} \rangle + e^{\imath \alpha_{1}} \vert s'_{\mathrm{A}} \rangle \vert s'_{\mathrm{B}} \rangle \right)\nonumber\\
&\hspace{22mm}+\frac{\eta_{2}}{\sqrt{2}}\left(\vert s_{\mathrm{A}} \rangle \vert s'_{\mathrm{B}} \rangle - e^{\imath \alpha_{2}} \vert s'_{\mathrm{A}} \rangle \vert s_{\mathrm{B}} \rangle \right)\Big],
\end{align}
where 
\begin{align*}
\eta_{1}&=(s_{\mathrm{A}}\cdot h_{\mathrm{A}})(s_{\mathrm{B}}\cdot h_{\mathrm{B}})+e^{\imath\alpha}(s_{\mathrm{A}}\cdot v_{\mathrm{A}})(s_{\mathrm{B}}\cdot v_{\mathrm{B}}),\\
\eta_{2}&=(s_{\mathrm{A}}\cdot h_{\mathrm{A}})(s'_{\mathrm{B}}\cdot h_{\mathrm{B}})+e^{\imath\alpha}(s_{\mathrm{A}}\cdot v_{\mathrm{A}})(s'_{\mathrm{B}}\cdot v_{\mathrm{B}}),
\end{align*}
and $\alpha_{i}$ is defined through the relation $\eta_{i}= \vert \eta_{i}\vert e^{\imath(\alpha-\alpha_{i})/2}$.  
Time delays resulting from \ac{PMD} in the fibers can now be described as
\begin{align}
\vert \psi_{\mathrm{PMD}} \rangle &= \frac{\eta_{1}}{\sqrt{2}}\vert f(t_{\mathrm{A}}-\frac{\tau_{\mathrm{A}}}{2},t_{\mathrm{B}}-\frac{\tau_{\mathrm{B}}}{2}) \rangle \otimes\vert s_{\mathrm{A}} s_{\mathrm{B}} \rangle +\nonumber\\
&\hspace{4.5mm}\frac{\eta_{2}}{\sqrt{2}}\vert f(t_{\mathrm{A}}-\frac{\tau_{\mathrm{A}}}{2},t_{\mathrm{B}}+\frac{\tau_{\mathrm{B}}}{2}) \rangle\otimes \vert s_{\mathrm{A}} s'_{\mathrm{B}} \rangle- \nonumber\\
&\hspace{4.5mm} \frac{\eta_{2}e^{\imath \alpha_{2}}}{\sqrt{2}}\vert f(t_{\mathrm{A}}+\frac{\tau_{\mathrm{A}}}{2},t_{\mathrm{B}}-\frac{\tau_{\mathrm{B}}}{2}) \rangle \otimes \vert s'_{\mathrm{A}} s_{\mathrm{B}} \rangle +\nonumber\\
&\hspace{4.5mm}\frac{\eta_{1}e^{\imath \alpha_{1}}}{\sqrt{2}}\vert f(t_{\mathrm{A}}+\frac{\tau_{\mathrm{A}}}{2},t_{\mathrm{B}}+\frac{\tau_{\mathrm{A}}}{2}) \rangle \otimes\vert s'_{\mathrm{A}} s'_{\mathrm{B}} \rangle.\label{eqn:Gen_PMD}
\end{align}

To account for the integration time of the photon detectors, the time modes of the two photons are to be traced out.
Then the polarization state of the two photons can be characterized by a density matrix for two qubits.
When written in the basis of $\vert s_{\mathrm{A}} s_{\mathrm{B}} \rangle$, $\vert s_{\mathrm{A}} s'_{\mathrm{B}} \rangle$, $\vert s'_{\mathrm{A}} s_{\mathrm{B}} \rangle$, and $\vert s'_{\mathrm{A}} s'_{\mathrm{B}} \rangle$,
the density matrix resulting from integration of time results is given by \eqref{eqn:Gen_Density}, in which
\begin{align}R(\tau_{\mathrm{A}},\tau_{\mathrm{B}})=\int\int\text{d}t_{\mathrm{A}}\text{d}t_{\mathrm{B}}f(t_{\mathrm{A}}+\tau_{\mathrm{A}},t_{\mathrm{B}}+\tau_{\mathrm{B}})f^{\dag}(t_{\mathrm{A}},t_{\mathrm{B}})
\label{eqn:RAB}
\end{align}
with the property that $R(0,0)=1$.

The approach above can also be applied to scenarios involving $\chi^{(3)}$ media,
which changes \eqref{eqn:Ep} and in turn \eqref{eqn:RAB}.
Since these changes have minor impact on the analytical results as well as the numerical findings in this paper, we omit the analysis for $\chi^{(3)}$ media to avoid redundancy.

\section{Local Rotation on One Photon is Sufficient for Alignment}
\label{pf_align}
We will first prove a lemma, and then show that as a special case of the lemma, local rotation on one of the photons can achieve
the alignment of the \ac{PSP} basis with the photon polarization basis.
\begin{Lem}[The basis of maximally entangled states] \label{lem:rotation}
$|\phi\rangle$  is a maximally entangled state of two qubits, and $\{|s\rangle, |s'\rangle\}$ is an  arbitrary basis of a qubit. Then there exists some basis of a qubit $\{|\tilde{s}\rangle, |\tilde{s}'\rangle\}$ such that
\begin{align}
|\phi\rangle= \frac{1}{\sqrt{2}}(|\tilde{s}s\rangle + |\tilde{s}'s'\rangle)\label{eqn:schmidt}
\end{align}
\end{Lem}
\proof
Express $|\phi\rangle$ in the basis of $\{|s\rangle, |s'\rangle\}$, i.e.,
\begin{align}
|\phi\rangle&= \alpha_{00} |ss\rangle + \alpha_{01} |ss'\rangle + \alpha_{10} |s's\rangle + \alpha_{11} |s's'\rangle\nonumber\\
& = (\alpha_{00} |s\rangle + \alpha_{10} |s'\rangle)\otimes |s\rangle + 
(\alpha_{01} |s\rangle + \alpha_{11} |s'\rangle) \otimes |s'\rangle.
\label{eqn:phi}
\end{align}
Denote $\M{A} =\begin{bmatrix}\alpha_{00}&\alpha_{01}\\\alpha_{10} & \alpha_{11}\end{bmatrix}$, and perform singular value decomposition on $\M{A}$
\begin{align*}
\M{A} = \M{U}\M{D}\M{V}
\end{align*}
where  $\M{U}$, $\M{V}$ are unitary matrices and $\M{D}$ is a diagonal matrix. Since $|\phi\rangle$ is a maximally entangled state of two qubits, all the singular values of $\M{A}$ must be $\frac{1}{\sqrt{2}}$.
Hence, $\M{D}=\frac{1}{\sqrt{2}}\mathbb{I}_2$, and $\M{A}$ can be rewritten as
\begin{align}
\M{A} = \frac{1}{\sqrt{2}}\M{U}\M{V} = \frac{1}{\sqrt{2}}\tilde{\M{U}}\label{eqn:A}.
\end{align}
Since $\M{U}$, $\M{V}$ are unitary matrices, so is $\tilde{\M{U}}$. Denote
\begin{align}
\begin{bmatrix}|\tilde{s}\rangle& |\tilde{s}'\rangle\end{bmatrix}  =  \begin{bmatrix}|s\rangle& |s'\rangle\end{bmatrix}\tilde{\M{U}}\label{eqn:s}
\end{align}
then since $\tilde{\M{U}}$ is unitary,  $\{|\tilde{s}\rangle, |\tilde{s}'\rangle\}$ is also a basis of a qubit.
Substitue \eqref{eqn:A} and \eqref{eqn:s} into \eqref{eqn:phi}, one can obtain \eqref{eqn:schmidt}.
This completes the proof.
\endproof

The photon source generates photon pairs whose polarization state is maximally entangled, i.e., 
\begin{align*}
|\phi\rangle=\frac{1}{\sqrt{2}}(\vert h_{\mathrm{A}} \rangle \vert h_{\mathrm{B}} \rangle + e^{\imath \alpha} \vert v_{\mathrm{A}} \rangle \vert v_{\mathrm{B}}\rangle ).
\end{align*} 
From Lemma~\ref{lem:rotation}, there exists some basis $\{|\tilde{s}_{\mathrm{A}}\rangle,|\tilde{s}'_{\mathrm{A}}\rangle\}$ such that $|\phi\rangle$ can be rewritten as
\begin{align}
|\phi\rangle=\frac{1}{\sqrt{2}}(\vert \tilde{s}_{\mathrm{A}} \rangle \vert s_{\mathrm{B}} \rangle + \vert \tilde{s}'_{\mathrm{A}} \rangle \vert s_{\mathrm{B}} \rangle ).\label{eqn:pol_phi}
\end{align}
From \eqref{eqn:pol_phi}, the polarization state prepared by the source can be viewed as a state in which the polarization basis of photon B is already aligned with the \ac{PSP} basis of the channel.
Hence, rotating photon A to align $\{|\tilde{s}_{\mathrm{A}}\rangle,|\tilde{s}'_{\mathrm{A}}\rangle\}$ with the \ac{PSP} basis $\{|s_{\mathrm{A}}\rangle,|s'_{\mathrm{A}}\rangle\}$ is sufficient to reduce the possible coincident arrival times of the photon pair to two.

\section{Proof of Theorem~\ref{thm:UBF}}
\label{pf_thm:UBF}
The two network nodes perform the following local unitary operations
\begin{align}
\begin{split}
\M{U}_{\mathrm{A}}&=\frac{|0\rangle+|1\rangle}{\sqrt{2}}\langle a| + \frac{|0\rangle-|1\rangle}{\sqrt{2}}\langle a'|,\\
\M{U}_{\mathrm{B}}&=\frac{|0\rangle+|1\rangle}{\sqrt{2}}\langle b| + e^{-\imath\theta}\frac{|0\rangle-|1\rangle}{\sqrt{2}}\langle b'|
\end{split}\label{eqn:UAB}
\end{align}
on a pair of qubits with density matrix $\V{\rho}$. The updated density matrix is given by
\begin{align}
\check{\V\rho}&=(\M{U}_{\mathrm A}\otimes\M{U}_{\mathrm B})\,\V{\rho}\,(\M{U}_{\mathrm A}\otimes\M{U}_{\mathrm B})^\dag\nonumber 
\\&= 
F|\Phi^+\rangle\langle\Phi^+| + (1-F)|\Psi^+\rangle\langle\Psi^+|\label{eqn:rho_transformed}
\end{align}
where
\begin{align*}
|\Phi^+\rangle &= \frac{1}{\sqrt{2}} (|00\rangle + |11\rangle)\\
|\Psi^+\rangle &= \frac{1}{\sqrt{2}} (|01\rangle + |10\rangle)
\end{align*}

The density matrix in \eqref{eqn:rho_transformed} has the structure of the density matrix in \cite[Eq.(6)]{RuaDaiWin:J18}, with $\alpha = \beta=\gamma=\delta =\frac{1}{\sqrt{2}}$.
Therefore, one can adopt \cite[Thm. 2]{RuaDaiWin:J18} and get 
\begin{align*}
F^*(\check{\V\rho}) = \frac{F^2}{F^2+(1-F)^2}.
\end{align*}
Moreover, since unitary operations are reversible, $F^*(\check{\V\rho})=F^*(\V\rho)$. This completes the proof.

\section{Proof of Theorem~\ref{thm:UBP}}
\label{pf_thm:UBP}
First prove that the proposed success probability is an upper bound, i.e.,
\begin{align}
P^*(\V{\rho}) \le F^2+(1-F)^2.
\end{align}
The statement will be proved by contradiction. 
Suppose the theorem does not hold, i.e., for some $\V\rho \in\mathcal{S}$ with $|R|>0$ there exists a distillation operation $\mathcal{D}$ such that 
\begin{align}
F_{\mathcal{D}}(\V{\rho}) &= \frac{F^2}{F^2+(1-F)^2}\label{eqn:FDopt}\\
P_{\mathcal{D}}(\V{\rho}) &> F^2+(1-F)^2 \label{eqn:PDoptplus}.
\end{align}

From \eqref{eqn:generalrho_2}, the spectrum decomposition of the joint density matrix of two qubit pairs is given by
\begin{align*}
\V{\rho}^{\mathrm{J}}&= F^2|\phi_1\phi_1\rangle\langle\phi_1\phi_1| + F(1-F) |\phi_1\phi_2\rangle\langle\phi_1\phi_2| \\ 
&\hspace{4.3mm}+(1-F)F |\phi_2\phi_1\rangle\langle\phi_2\phi_1| + (1-F)^2|\phi_2\phi_2\rangle\langle\phi_2\phi_2|
\end{align*}
Define
\begin{align*}
\M{V}_{nm} &= \mathrm{tr}_{3,4}\{\mathcal{D}\{|\phi_n\phi_m\rangle\langle\phi_n\phi_m|\}\}\\
f_{nm} & = \langle \Phi^+| \M{V}_{nm}  |\Phi^+\rangle \\
p_{nm} & = \mathrm{tr}\{\M{V}_{nm}\}
\end{align*}
where $n,m\in\{1,2\}$. 
As along as $\mathcal{D}$ is a valid quantum operation, $\M{V}_{nm}$ must be a positive semidefinite matrix  with trace no greater than 1.
Therefore,
\begin{align}
0\le f_{nm} \le p_{nm} \le 1. \label{eqn:fp}
\end{align}   
It is straight forward that
\begin{align}
F_{\mathcal{D}}(\V{\rho}) &= \frac{F^2f_{11}+F(1-F)(f_{12} + f_{21}) +(1-F)^2f_{22}}{F^2p_{11}+F(1-F)(p_{12} + p_{21}) +(1-F)^2p_{22}}\label{eqn:FD}\\
P_{\mathcal{D}}(\V{\rho}) &= F^2p_{11}+F(1-F)(p_{12} + p_{21}) +(1-F)^2p_{22} \label{eqn:PD}.
\end{align}
Combining \eqref{eqn:PDoptplus} and \eqref{eqn:PD}, and noticing that $p_{nm} \le 1$, it can be derived that
\begin{align}
p_{12} + p_{21} >0\label{eqn:nonzerop}
\end{align}

Denote
\begin{align*}
S(F) & = F^2f_{11}+F(1-F)(f_{12} + f_{21}) +(1-F)^2f_{22}\\
N(F) & = F^2(p_{11}-f_{11})+F(1-F)(p_{12} + p_{21}-f_{12} - f_{21}) 
\\&\hspace{4.3mm}+(1-F)^2(p_{22}-f_{22})
\end{align*} 
Then from \eqref{eqn:FDopt} and \eqref{eqn:FD}
\begin{align}
&F_{\mathcal{D}}(\V{\rho})=\frac{S(F)}{S(F)+N(F)}=\frac{F^2}{F^2+(1-F)^2}\nonumber\\
\Rightarrow& \frac{N(F)}{S(F)} = \frac{(1-F)^2}{F^2}\label{eqn:SFratio} 
\end{align}

$F>\frac{1}{2}$ as $|R|>0$. Hence, one can construct another density matrix $\tilde{\V{\rho}}$ satisfying \eqref{eqn:generalrho_2}, with a different 
$\tilde{F}\in(\frac{1}{2},F)$. By repeating the analysis above, it can be derived that
\begin{align}
F_{\mathcal{D}}(\tilde{\V{\rho}})=\frac{S(\tilde{F})}{S(\tilde{F})+N(\tilde{F})}=\frac{1}{1+\frac{N(\tilde{F})}{S(\tilde{F})}}\label{eqn:FDtilde} 
\end{align}
From \eqref{eqn:fp} and \eqref{eqn:nonzerop}, if $f_{12} + f_{21}=p_{12} + p_{21} >0$, then
\begin{align}
S(\tilde{F}) &=  \frac{\tilde{F}^2}{F^2}\Big(F^2f_{11}+ \frac{F^2}{\tilde{F}}(1-\tilde{F})(f_{12} + f_{21}) \nonumber\\
                       & \hspace{4.3mm}+ \frac{F^2}{\tilde{F}^2}(1-\tilde{F})^2f_{22}\Big)\nonumber\\
                  &>  \frac{\tilde{F}^2}{F^2}\Big(F^2f_{11}+ F(1-F)(f_{12} + f_{21}) 
                        + (1-F)^2f_{22}\Big)\nonumber\\
                  &=  \frac{\tilde{F}^2}{F^2} S(F)\label{eqn:SFtilde}\\ 
N(\tilde{F}) &=  \frac{(1-\tilde{F})^2}{(1-F)^2}\Big(\frac{(1-F)^2}{(1-\tilde{F})^2}\tilde{F}^2(p_{11}-f_{11})\nonumber\\
                   &\hspace{4.3mm}+\tilde{F}\frac{(1-F)^2}{(1-\tilde{F})}(p_{12} + p_{21}-f_{12} - f_{21})\nonumber\\ 
                   &\hspace{4.3mm}+(1-F)^2(p_{22}-f_{22})\Big)\nonumber\\
                  &\le\frac{(1-\tilde{F})^2}{(1-F)^2}\Big(F^2(p_{11}-f_{11})\nonumber\\
                   &\hspace{4.3mm}+F(1-F)(p_{12} + p_{21}-f_{12} - f_{21}) \nonumber\\
                   &\hspace{4.3mm}+(1-F)^2(p_{22}-f_{22})\Big)\nonumber \\
                   & =\frac{(1-\tilde{F})^2}{(1-F)^2}N(F) \label{eqn:NFtilde}                                  
\end{align}

Substituting  \eqref{eqn:SFratio}, \eqref{eqn:SFtilde}, and \eqref{eqn:NFtilde} into \eqref{eqn:FDtilde}, one can get
\begin{align*}
F_{\mathcal{D}}(\tilde{\V{\rho}}) > \frac{\tilde{F}^2}{\tilde{F}^2+(1-\tilde{F})^2}
\end{align*}
which leads to 
\begin{align}
F^*(\tilde{\V{\rho}})\ge F_{\mathcal{D}}(\tilde{\V{\rho}}) > \frac{\tilde{F}^2}{\tilde{F}^2+(1-\tilde{F})^2}.\label{eqn:Ftildeoptplus}
\end{align}
However, \eqref{eqn:Ftildeoptplus} contradicts with \eqref{eqn:Fidelity_Upper}.

Otherwise, if $p_{12} + p_{21}>f_{12} + f_{21}\ge 0$, one can use similar analysis and get
\begin{align*}
S(\tilde{F}) &\ge  \frac{\tilde{F}^2}{F^2} S(F)\\ 
N(\tilde{F}) &<  \frac{(1-\tilde{F})^2}{(1-F)^2}N(F)                                 
\end{align*}
which also lead to a contradiction between \eqref{eqn:Ftildeoptplus} and \eqref{eqn:Fidelity_Upper}. 
This contradiction shows that success probability given in \eqref{eqn:Prob_Upper} is indeed an upper bound.

The achievability of  \eqref{eqn:Prob_Upper} will be proved constructively with the \ac{QED} algorithm to be proposed. Please refer to Section~\ref{sec:alg} for details.

\section{Proof of Theorem~\ref{thm:dis_perf}}
\label{pf_thm:dis_perf}

From \eqref{eqn:rho_transformed}, after the first step of the algorithm, the joint density matrix of two qubit pairs is given by
\begin{align*}
{\V\rho}_{\mathrm{J}}&=\M{P}\check{\V\rho}\otimes\check{\V\rho}\,\M{P}^\dag\\
&=F^2|\Omega^{(1)}\rangle\langle\Omega^{(1)}| +
F(1-F)\big(|\Omega^{(2)}\rangle\langle\Omega^{(2)}|
+ |\Omega^{(3)}\rangle\langle\Omega^{(3)}|\big)\\
&\hspace{3.7mm}+ (1-F)^2|\Omega^{(4)}\rangle\langle\Omega^{(4)}|
\end{align*}
where $\M{P}$ is the permutation operator that switches the second and third qubits, and
\begin{align}
\begin{array}{*{8}{r@{\;}}r@{}r}
|\Omega^{(1)}\rangle&=&&
\frac{1}{2}|0000 \rangle&+&
\frac{1}{2}|0101 \rangle
\vspace{1mm}\\
&&+&
\frac{1}{2}|1010 \rangle&+&
\frac{1}{2}|1111 \rangle
\vspace{1mm}\\
|\Omega^{(2)}\rangle&=&&
\frac{1}{2}|0001 \rangle&+&
\frac{1}{2}|0100 \rangle
\vspace{1mm}\\
&&+&
\frac{1}{2}|1011 \rangle&+&
\frac{1}{2}|1110\rangle
\vspace{1mm}\\
|\Omega^{(3)}\rangle&=&&
\frac{1}{2}|0010 \rangle&+&
\frac{1}{2}|0111\rangle
\vspace{1mm}\\
&&+&
\frac{1}{2}|1000\rangle&+&
\frac{1}{2}|1101\rangle
\vspace{1mm}\\
|\Omega^{(4)}\rangle&=&&
\frac{1}{2}|0011 \rangle&+&
\frac{1}{2}|0110\rangle
\vspace{1mm}\\
&&+&
\frac{1}{2}|1001\rangle&+&
\frac{1}{2}|1100\rangle&&.
\end{array}
\nonumber
\end{align}

In the first round of distillation, after both nodes perform the CNOT operation, the joint density matrix of two qubit pairs becomes
\begin{align}
\check{\V\rho}_{\mathrm{J}}
&=F^2|\check\Omega^{(1)}\rangle\langle\check\Omega^{(1)}| +
F(1-F)\big(|\check\Omega^{(2)}\rangle\langle\check\Omega^{(2)}|
\nonumber
\\&\hspace{4.3mm}+ |\check\Omega^{(3)}\rangle\langle\check\Omega^{(3)}|\big)
+ (1-F)^2|\check\Omega^{(4)}\rangle\langle\check\Omega^{(4)}|\label{eqn:jointBCNOT}
\end{align}
where\begin{align}
\begin{array}{*{8}{r@{\;}}r@{}r}
|\check\Omega^{(1)}\rangle&=&&
\frac{1}{2}|0000\rangle&+&
\frac{1}{2}|0101\rangle
\vspace{1mm}\\
&&+&
\frac{1}{2}|1111\rangle&+&
\frac{1}{2}|1010 \rangle\vspace{1mm}
\\
|\check\Omega^{(2)}\rangle&=&&
\frac{1}{2}|0001\rangle&+&
\frac{1}{2}|0100\rangle
\vspace{1mm}\\
&&+&
\frac{1}{2}|1110\rangle&+&
\frac{1}{2}|1011\rangle
\vspace{1mm}\\
|\check\Omega^{(3)}\rangle&&=&
\frac{1}{2}|0011\rangle&+&
\frac{1}{2}|0110\rangle
\vspace{1mm}\\
&&+&
\frac{1}{2}|1100\rangle&+&
\frac{1}{2}|1001\rangle\vspace{1mm}\\
|\check\Omega^{(4)}\rangle&&=&
\frac{1}{2}|0010\rangle&+&
\frac{1}{2}|0111\rangle
\vspace{1mm}\\
&&+&
\frac{1}{2}|1101 \rangle&+&
\frac{1}{2}|1000\rangle&.
\end{array}
\nonumber
\end{align}

From \eqref{eqn:jointBCNOT}, 
if both measurement results correspond to $|0\rangle\langle 0|$, the (unnormalized) density matrix of the source qubit pair is given by
\begin{align}
\V{\rho}_{00}&=(\mathbb{I}_2\otimes\langle 0|\otimes\mathbb{I}_2\otimes\langle 0
|)\,\check{\V\rho}_{\mathrm{J}}\,(\mathbb{I}_2\otimes|0\rangle\otimes\mathbb{I}_2\otimes|0\rangle)\nonumber
\\&=\frac{1}{2}\big(F^2|\Phi^+\rangle\langle\Phi^+| + (1-F)^2|\Psi^+\rangle\langle\Psi^+|  \big).
\label{eqn:case00}
\end{align}
Similarly, if both measurement results correspond to $|1\rangle\langle1|$, the (unnormalized) density matrix of the source qubit pair is given by
\begin{align}
\V\rho_{11}&=(\mathbb{I}_2\otimes\langle 1|\otimes\mathbb{I}_2\otimes\langle 1
|)\,\check{\V\rho}_{\mathrm{J}}\,(\mathbb{I}_2\otimes|1\rangle\otimes\mathbb{I}_2\otimes|1\rangle)\nonumber
\\&=\frac{1}{2}\big(F^2|\Phi^+\rangle\langle\Phi^+| + (1-F)^2|\Psi^+\rangle\langle\Psi^+|  \big).\label{eqn:case11}
\end{align}

From \eqref{eqn:case00}, and \eqref{eqn:case11}, the probability of preserving the source qubit pair is
\begin{align}
P=\mathrm{tr}\{\V\rho_{00}+\V\rho_{11}\}=F^2+(1-F)^2\label{eqn:P}
\end{align}
the fidelity of the kept qubit pairs is
\begin{align}
F_1=\frac{\frac{1}{2}F^2+\frac{1}{2}F^2}{P}=\frac{F^2}{F^2+(1-F)^2}\label{eqn:F1}
\end{align}
and the density matrix of the kept qubit pair can be written as
\begin{align}
\V\rho^{(1)}=\frac{\V\rho_{00}+\V\rho_{11}}{P}=F_1|\Phi^+\rangle\langle\Phi^+| +(1-F_1)|\Psi^+\rangle\langle\Psi^+|.\label{eqn:rho1}
\end{align}

With \eqref{eqn:P} and \eqref{eqn:F1}, the proof for the first round of distillation is complete.
For the following rounds of distillations, one can take \eqref{eqn:rho1} as input, and repeat the analysis in \eqref{eqn:jointBCNOT}--\eqref{eqn:F1}. 
This competes the proof.

\end{document}